\documentclass[final, journal, letterpaper]{IEEEtran}

\usepackage[utf8]{inputenc}
\usepackage{graphicx}
\usepackage{stfloats}
\usepackage{amssymb}
\usepackage{amsmath}
\usepackage{color}
\usepackage{theorem}
\usepackage{pifont}
\usepackage{array}
\usepackage{hyperref}
\hypersetup{colorlinks=false,pdfborderstyle={/S/U/W 1}}
\usepackage{cite}
\usepackage{url}
\usepackage{bbm}
\usepackage{relsize}
\usepackage{nicefrac}
\usepackage{pgfplots}
\usepackage{subcaption}

\usepgfplotslibrary{external}
\tikzsetexternalprefix{tikzfig/}
 \newtheorem{theorem}{Theorem}
 \newtheorem{definition}{Definition}
 \newtheorem{lemma}{Lemma}
 \newtheorem{corollary}{Corollary}
\newtheorem{problem}{Problem}
 \newtheorem{remark}{Remark}
 \newtheorem{proposition}{Proposition}

 \newcommand{\qed}{\hfill \mbox{\raggedright \rule{.07in}{.1in}}}
 \newenvironment{proof}{\vspace{1ex}\noindent{\bf Proof}\hspace{0.5em}}
 	{\hfill\qed\vspace{1ex}} 

\DeclareMathOperator{\rk}{rk}

\DeclareMathOperator{\defi}{def}

\newcommand{\defeq}{\overset{\defi}{=}}
\newcommand{\Fqm}{\mathbb F_{q^m}}
\newcommand{\Fqn}{\mathbb F_{q^n}}

\newcommand{\Fq}{\mathbb F_{q}}

\newcommand{\dhalffrac}{\left\lfloor \frac{d-1}{2}\right\rfloor}

\newcommand{\dhalfnicefrac}{\left\lfloor \nicefrac{(d-1)}{2}\right\rfloor}

\newcommand{\taujohnson}{n-\sqrt{n(n-d)}}
\newcommand{\Ball}[2]{\mathcal{B}_{#1}(#2)}
\newcommand{\Sphere}[2]{\mathcal{S}_{#1}(#2)}

\newcommand{\myspace}[1]{\mathcal{#1}}
\newcommand{\Rowspace}[1]{\mathcal{R}_q(#1)}
\newcommand{\Colspace}[1]{\mathcal{C}_q(#1)}
\newcommand{\Grassm}[1]{\mathcal{G}_q(#1)}
\newcommand{\Projspace}{\Fq^n} 
\newcommand{\Subspacedist}[1]{d_S(#1)}

\newcommand{\maxCardinalityRank}[1]{\mycode{A}_{q^m}^R\left(#1\right)}
\newcommand{\maxCardinalitySub}[1]{\mycode{A}^S_{q}\left(#1\right)}
\newcommand{\List}{\mathcal L}

\newcommand{\mycode}[1]{\ensuremath{\mathsf{#1}}}

\newcommand{\code}[1]{\ensuremath{(#1)_R}}

\newcommand{\MRDlinear}[1]{\ensuremath{\mycode{MRD}[#1]}}
\newcommand{\Gab}[1]{\ensuremath{\mycode{Gab}[#1]}}

\newcommand{\CRC}[1]{\ensuremath{\mycode{CR}_{q^m}(#1)}}
\newcommand{\CDC}[1]{\ensuremath{\mycode{CD}_q(#1)}}

\renewcommand{\a}{\mathbf a}
\renewcommand{\b}{\mathbf b}
\renewcommand{\c}{\mathbf c}

\renewcommand{\r}{\mathbf r}

\newcommand{\x}{\mathbf x}

\newcommand{\A}{\mathbf A}
\newcommand{\B}{\mathbf B}
\newcommand{\C}{\mathbf C}
\newcommand{\D}{\mathbf D}

\newcommand{\G}{\mathbf G}
\newcommand{\I}{\mathbf I}
\renewcommand{\H}{\mathbf H}

\newcommand{\R}{\mathbf R}
\renewcommand{\S}{\mathbf S}
\newcommand{\X}{\mathbf X}
\newcommand{\Y}{\mathbf Y}

\newcommand{\0}{\mathbf 0}

\newcommand{\quadbinom}[2]{\ensuremath{
{#1
\brack
#2}
}}

\begin{document}

\title{Bounds on List Decoding of Rank-Metric Codes}

\IEEEoverridecommandlockouts

\author{{Antonia Wachter-Zeh, \IEEEmembership{Student Member, IEEE}
\thanks{Parts of 
this work were presented at the \textit{Thirteenth International Workshop on Algebraic and Combinatorial Coding Theory (ACCT)}, June 2012, Pomorie, Bulgaria \cite{Wachterzeh-BoundsListDecodingGabidulin_conf}
and at the \textit{IEEE International Symposium on Information Theory (ISIT)}, July 2013, Istanbul, Turkey \cite{Wachterzeh-BoundsRankMetricListDecoding_ISIT2013_conf}. 
This work has been supported by the German Research Council (DFG) under grant Bo~867/21-1.

A. Wachter-Zeh is with the Institute of Communications Engineering, University of Ulm, D-89081 Ulm, Germany
and Institut de Recherche Mathémathique de Rennes (IRMAR), Université de Rennes 1, 35042 Rennes Cedex, France
(e-mail: antonia.wachter@uni-ulm.de).
}}}

\maketitle

\begin{abstract}
So far, there is no polynomial-time list decoding algorithm (beyond half the minimum distance) for  Gabidulin codes. These codes can be seen as the rank-metric equivalent of Reed--Solomon codes.
In this paper, we provide bounds on the list size of rank-metric codes in order to understand whether polynomial-time list decoding is possible or whether it works only with exponential time complexity.
Three bounds on the list size are proven. The first one is a lower exponential bound for Gabidulin codes and shows that for these codes no polynomial-time
list decoding beyond the Johnson radius exists. 
Second, an exponential upper bound is derived, which holds for any rank-metric code of length $n$ and minimum rank distance $d$.
The third bound proves that there exists a rank-metric code over $\Fqm$ of length $n \leq m$ such that the list size is exponential in the length for any radius 
greater than \emph{half the minimum rank distance}. 
This implies that there cannot exist a \emph{polynomial} upper bound depending only on $n$ and $d$ similar to the Johnson bound in Hamming metric.
All three rank-metric bounds reveal significant differences to bounds for codes in Hamming metric.
\end{abstract}

\begin{keywords}
Rank-metric codes, Gabidulin codes, list decoding, constant-rank codes
\end{keywords}

\section{Introduction}\label{sec:introduction}
\IEEEPARstart{R}{ank-metric} codes lately attract more and more attention due to their possible application to error control in random linear network coding \cite{koetter_kschischang,silva_rank_metric_approach,GabidulinPilipchukBossert-CorrectingErasuresAndErrorsInRandomNetworkCoding_2010}. 
A code in rank metric can be considered as a set of $m \times n$ matrices over a finite field $\Fq$ or equivalently as a set of vectors of length $n$ over an extension field $\Fqm$ of $\Fq$.
The rank weight of a word is the rank of its matrix representation and the rank distance between two matrices is the rank of their difference.

\emph{Gabidulin} codes can be seen as the rank metric equivalent to Reed--Solomon codes and were introduced by Delsarte \cite{Delsarte_1978}, Gabidulin \cite{Gabidulin_TheoryOfCodes_1985}, and Roth \cite{Roth_RankCodes_1991}. 
They can be defined by evaluating degree-restricted \emph{linearized polynomials}, which were introduced by Ore \cite{Ore_OnASpecialClassOfPolynomials_1933,Ore_TheoryOfNonCommutativePolynomials_1933}. 
Additionally to the definition as evaluation codes, the similarities between Gabidulin and Reed--Solomon codes go further.
There are several algorithms for (unique) decoding of Gabidulin codes up to half the minimum rank distance,
which have a famous equivalent for Reed--Solomon codes: the algorithm by Roth \cite{Roth_RankCodes_1991} and similarly by Gabidulin \cite{Gabidulin1992Fast} solving a system of equations as the Peterson algorithm, 
a method based on the linearized Euclidean algorithm \cite{Gabidulin_TheoryOfCodes_1985}, Paramonov and Tretjakov's and Richter and Plass' Berlekamp--Massey-like linearized shift-register synthesis \cite{Paramonov_Tretjakov_BMA_1991,Richter_RankCodes_2004,RichterPlass_DecodingRankCodes_2004,Sidorenko2011Linearized},
Loidreau's Welch--Berlekamp-like algorithm \cite{Loidreau_AWelchBerlekampLikeAlgorithm_2006} and many more \cite{GabidulinPilipchuck_ErrorErasureRankCodes_2008,SilvaKschischang-FastEncodingDecodingGabidulin-2009,Xie2011General,WachterAfanSido-FastDecGabidulin_DCC_journ}.

A \emph{list decoding} algorithm returns the list of \emph{all} codewords in distance at most $\tau$ from any given word. 
The idea of list decoding was introduced by Elias \cite{elias_list_1957} and Wozencraft \cite{wozencraft_list_1958}.
In Hamming metric, the \emph{Johnson upper bound} \cite{Johnson1962New,Johnson1963Improved,Bassalygo1965New,Guruswami_ListDecodingofError-CorrectingCodes_1999,Guruswami_ALGORITHMICRESULTSINLISTDECODING_2007} shows that the size of this list is polynomial in $n$ when $\tau$ is less than the so-called Johnson radius
$\tau_J = n- \sqrt{n(n-d_H)}$ for \emph{any} code of length $n$ and minimum Hamming distance $d_H$.
Although this fact has been known since the 1960s, a polynomial-time list decoding algorithm for Reed--Solomon codes up 
to the Johnson radius was found not earlier than 1999 by Guruswami and Sudan \cite{Guruswami-Sudan:IEEE_IT1999} as a generalization of the Sudan algorithm \cite{Sudan:JOC1997}.
Further, in Hamming metric, it can be shown that there exists a code such that the list size becomes \emph{exponential} in $n$ when $\tau$ is at least the Johnson radius \cite{Goldreich-Rubinfeld-Sudan:DM2000}, \cite[Chapter~4]{Guruswami_ListDecodingofError-CorrectingCodes_1999}.
It is not known whether such an exponential list beyond the Johnson radius also exists for Reed--Solomon codes. There are several articles, 
which show an exponential behavior of the list size for Reed--Solomon codes only for a radius rather greater than $\tau_J$ (see e.g. Justesen and H\o{}holdt \cite{Justesen2001Bounds} and Ben-Sasson, Kopparty and Radhakrishnan \cite{BenSasson2010Subspace}).

However, for Gabidulin codes, so far there exists \emph{no} polynomial-time list decoding algorithm and it is not even known whether it can exist or not.
The contributions by Mahdavifar and Vardy and by Guruswami and Xing provide list decoding algorithms for special classes of Gabidulin codes and subcodes of Gabidulin codes \cite{Mahdavifar2010Algebraic,Mahdavifar2012Listdecoding,GuruswamiXing-ListDecodingRSAGGabidulinSubcodes_2012}.

On the one hand, a lower bound on the maximum list size, which is exponential in the length $n$ of the code, rules out the possibility of polynomial-time list decoding since already writing down the list has exponential complexity.
On the other hand, a polynomial upper bound---as the Johnson bound for Hamming metric---shows that polynomial-time list decoding algorithms might exist.

In this contribution, we investigate bounds on list decoding rank-metric codes in general and Gabidulin codes in particular.
We derive three bounds on the maximum list size when decoding rank-metric codes.
In spite of the numerous similarities between Hamming metric and rank metric and even more between Reed--Solomon and Gabidulin codes, 
all three bounds show a strongly different behavior for rank-metric codes. 
The first bound is a \emph{lower bound for Gabidulin codes} of length $n$ and minimum rank distance $d$, which proves (for $n=m$) an exponential list size if the radius is at least the Johnson radius $\tau_J = n- \sqrt{n(n-d)}$.
The second bound is an \emph{exponential upper bound for any rank-metric code}, which provides no conclusion about polynomial-time list decodability.
Finally, the third bound shows that there exists a rank-metric code over $\Fqm$ of length $n \leq m$ such that the list size is \emph{exponential} in the length $n$ when the decoding 
radius is greater than \emph{half the minimum distance}. For these codes, hence, no polynomial-time list decoding can exist. Moreover, it shows that 
purely as a function of the length $n$ and the minimum rank distance $d$, there cannot exist a polynomial upper bound for an arbitrary code similar to the Johnson bound for Hamming metric.

This paper is structured as follows.
Section~\ref{sec:preliminaries} introduces notations about finite fields and rank-metric codes. 
Moreover, we give some useful lemmas, state the problem and show connections between constant-rank and constant-dimension codes. 
In Section~\ref{sec:bound_gabidulin}, the lower bound for list decoding of Gabidulin codes is derived using the evaluation of linearized polynomials.
Section~\ref{sec:bounds} first explains how the list of codewords is connected to a constant-rank code and provides the upper bound, which holds for any rank-metric code.
Second, we derive the existence of a rank-metric code with exponential list size beyond half the minimum distance.
Finally, in Section~\ref{sec:interpretation}, we interpret the new bounds and explain the differences between the bounds for Hamming and rank metric.

\section{Preliminaries}\label{sec:preliminaries}
\subsection{Finite Fields and Subspaces}
Let $q$ be a power of a prime, and let us denote by $\Fq$ the finite field of order $q$ and by $\Fqm$ its extension field of degree $m$. 
We use $\Fq^{s \times n}$ to denote the set of all $s\times n$ matrices over $\Fq$ and 
$\Fqm^n =\Fqm^{1 \times n}$ for the set of all {row} vectors of length $n$ over $\Fqm$. 
Therefore, $\Fq^n$ denotes the vector space of dimension $n$ over $\Fq$. 
The \emph{Grassmannian} of dimension $r$ is the set of all subspaces of $\Fq^n$ of dimension $r \leq n$ and is denoted by $\Grassm{n,r}$. 
The cardinality of $\Grassm{n,r}$ is the so-called Gaussian binomial, calculated by
\begin{equation*}
\big|\Grassm{n,r}\big|=\quadbinom{n}{r} \defeq \prod\limits_{i=0}^{r-1} \frac{q^n-q^i}{q^r-q^i},
\end{equation*}
with the upper and lower bounds (see e.g. \cite[Lemma~4]{koetter_kschischang})
\begin{equation}\label{eq:bounds_gaussian_binomial}
q^{r(n-r)}\leq \quadbinom{n}{r} \leq 4 q^{r(n-r)}.
\end{equation}

For two subspaces $\myspace{U},\myspace{V}$ in $\Projspace$, we denote by $\myspace{U}+\myspace{V}$ the smallest subspace containing the union of $\myspace{U}$ and $\myspace{V}$. 
The \emph{subspace distance} between $\myspace{U},\myspace{V}$ in $\Projspace$ is defined by
\begin{align*}
\Subspacedist{\myspace{U},\myspace{V}}&=\dim(\myspace{U}+\myspace{V})-\dim(\myspace{U}\cap \myspace{V})\\
&=2 \dim(\myspace{U}+\myspace{V})-\dim(\myspace{U})-\dim(\myspace{V}).
\end{align*}
It can be shown that the subspace distance is indeed a metric (see e.g. \cite{koetter_kschischang}).

A \emph{subspace code} is a non-empty subset of subspaces of $\Fq^n$ and has minimum subspace distance $d_S$, when all subspaces in the code have subspace distance at least $d_S$.
The codewords of a subspace code are therefore subspaces.
A \emph{constant-dimension code} of dimension $r$, cardinality $M$ and minimum subspace distance $d_S$ is a subset of $\Grassm{n,r}$, i.e., it is a special subspace code and is
denoted by $\CDC{n,M,d_S,r}$. 

The maximum cardinality of a constant-dimension code for fixed parameters $n, d_S, r$ is denoted by $\maxCardinalitySub{n,d_S,r}$.

\subsection{Rank-Metric Codes}

For a given basis of $\Fqm$ over $\Fq$, there exists a one-to-one mapping for each vector $\mathbf x \in \mathbb{F}_{q^m}^n$ on a matrix $\mathbf X \in \Fq^{m \times n}$. 
Let $\rk(\mathbf x)$ denote the (usual) rank of $\mathbf X$ over $\mathbb{F}_{q}$ and let 
$\Rowspace{\X}$, $\Colspace{\X}$ denote the row  and column space of $\X$ over $\Fq$. 
The right kernel of a matrix is denoted by $\ker(\x) = \ker(\X)$. 
The rank-nullity theorem states that for an $m \times n$ matrix, if $\dim \ker(\x) = t$, then $\dim \Colspace{\X}=\rk(\x) = n-t$.
Throughout this paper, we use the notation as vector (e.g. from $\Fqm^n$) or matrix (e.g. from $\Fq^{m \times n}$) equivalently, whatever is more convenient.

The \emph{minimum rank distance} $d_R$ of a block code $\mycode{C}$ is defined by 
\begin{equation*}
d_R = \min \big\lbrace \rk(\mathbf c_1-\c_2) \; : \; \mathbf c_1,\c_2 \in \mycode{C}, \mathbf c_1 \neq \mathbf c_2\big\rbrace. 
\end{equation*}

Let an $\code{n,M, d}$ code $\mycode{C}$ over $\Fqm$ denote a code in rank metric (not necessarily linear) of cardinality $M$ and minimum rank distance $d_R=d$. 
Its codewords are in $\Fqm^n$ or equivalently represented as matrices in $\Fq^{m\times n}$. 
W.l.o.g. we assume throughout this paper that $n\leq m$. If this is not the case, we consider the transpose of all matrices such that $n \leq m$ holds.
We call $n$ the \emph{length} of such a block code in rank metric over $\Fqm$.

The cardinality $M$ of an $\code{n,M, d}$ code over $\Fqm$ with $n \leq m$ is restricted by a \emph{Singleton}-like upper bound (see \cite{Delsarte_1978,Gabidulin_TheoryOfCodes_1985,Roth_RankCodes_1991}):
\begin{equation}\label{eq:Singleton_like_rank}
M \leq q^{\min\{n(m-d+1),\; m(n-d+1)\}} =q^{m(n-d+1)}.
\end{equation}
For linear codes of length $n \leq m$ and dimension $k$, this implies that $d_R \leq n-k+1$.
If the cardinality of a code fulfills \eqref{eq:Singleton_like_rank} with equality, the code is called a \emph{maximum rank distance} (MRD) code. 
A linear MRD code over $\Fqm$ of length $n \leq m$, dimension $k$ and minimum rank distance $d_R=n-k+1$ is denoted by $\MRDlinear{n,k}$ and has cardinality $M = q^{mk}$.

A special class of rank-metric codes are \emph{constant-rank codes}. Such a $\CRC{n,M,d,r}$ constant-rank code of length $n\leq m$ and minimum rank distance $d_R = d$ over $\Fqm$ is an $\code{n,M, d}$ code in rank metric, where all codewords have the same rank $r$.
The maximum cardinality of a constant-rank code for fixed parameters $n, d_R, r$ is denoted by $\maxCardinalityRank{n,d_R,r}$.

Further, $\Ball{\tau}{\a}$ denotes a ball of radius $\tau$ in rank metric around a word $\a\in \Fqm^n$ and $\Sphere{\tau}{\a}$ denotes a sphere in rank metric of radius $\tau$ around the word $\a$. 
The cardinality of $\Sphere{\tau}{\a}$ is the number of $m\times n$ matrices in $\Fq$, which have rank distance \emph{exactly} $\tau$ from $\a$ and the cardinality of a ball of radius $\tau$ is the number of $m\times n$ matrices in $\Fq$, which have rank distance \emph{less than or equal to} $\tau$. Therefore, (see e.g. \cite{Gadouleau2008Packing}):
\begin{equation*}
|\Ball{\tau}{\a}| = \sum\limits_{i=0}^{\tau} |\Sphere{i}{\a}| = \sum\limits_{i=0}^{\tau} \quadbinom{m}{i} \prod\limits_{j=0}^{i-1} (q^n-q^j).
\end{equation*}
The volumes of $\Ball{\tau}{\a}$ and $\Sphere{\tau}{\a}$ are independent of the choice of their center.

\subsection{Gabidulin Codes}
Gabidulin codes \cite{Delsarte_1978,Gabidulin_TheoryOfCodes_1985,Roth_RankCodes_1991} are a special class of MRD codes and 
are often considered as the analogs of Reed--Solomon codes in rank metric.
In order to define Gabidulin codes as evaluation codes, we give some basic properties of linearized polynomials \cite{Ore_OnASpecialClassOfPolynomials_1933,Ore_TheoryOfNonCommutativePolynomials_1933,Lidl-Niederreiter:FF1996}.

Let us denote the $q$-power by $x^{[i]}= x^{q^i}$ for any integer $i$. 
A \emph{linearized polynomial} 
over $\Fqm$ has the form
\begin{equation*}
f(x) = \sum_{i=0}^{d_f} f_i x^{[i]}, 
\end{equation*}
with $f_i \in \Fqm$. 
If the coefficient $f_{d_f}\neq 0$, we call $d_f \defeq \deg_q f(x)$ the \textit{q-degree} of $f(x)$. 
For all $\alpha_1,\alpha_2 \in \Fq$ and all $a,b \in \Fqm$, the following holds: 
\begin{equation*}
f(\alpha_1 a+\alpha_2 b) = \alpha_1 f(a)+\alpha_2 f(b).
\end{equation*}
The (usual) addition and the non-commutative composition $f(g(x))$ (also called \emph{symbolic product}) convert the set of linearized polynomials into a non-commutative ring with the identity element $x^{[0]}=x$. In the following, all polynomials are linearized polynomials. 

A \emph{Gabidulin code} can be defined by the evaluation of degree-restricted linearized polynomials as follows.
\begin{definition}[Gabidulin Code, \cite{Gabidulin_TheoryOfCodes_1985}]\label{def:gabidulin_code}
A linear Gabidulin code $\Gab{n,k}$ over $\Fqm$ of length $n\leq m$ and dimension $k \leq n$ is the set of all words, 
which are the evaluation of a $q$-degree-restricted linearized polynomial $f(x)$:
\begin{equation*}
\Gab{n,k}\! \defeq\!\Big\lbrace\! \left(f(\alpha_0) \; f(\alpha_1) \; \dots f(\alpha_{n-1})\right) : \deg_q f(x) < k)\!\Big\rbrace,
\end{equation*}
where the fixed elements $\alpha_0,\alpha_1, \dots, \alpha_{n-1} \in \Fqm$ are linearly independent over $\Fq$. 
\end{definition}
It can be shown that Gabidulin codes are MRD codes, i.e., $d_R=n-k+1$ (see \eqref{eq:Singleton_like_rank}).

\subsection{Connections between Constant-Dimension and Constant-Rank Codes}
Gadouleau and Yan showed in \cite{Gadouleau2010ConstantRank} connections between constant-dimension and constant-rank codes. In this subsection, we recall and generalize some of their results, since we use them in the next sections for bounding the list size.

The first lemma is a well-known algebraic fact, called rank decomposition.
\begin{lemma} [Rank Decomposition, {\cite[Theorem~3.13]{Stewart-MatrixAlgorithms_Book_1998}}]\label{lem:rank_decomp} 

Let a matrix $\X \in \Fq^{m\times n}$ of rank $r$ be given.
Then, there exist full rank matrices $\G \in \Fq^{r\times m}$ and $\H \in \Fq^{r\times n}$  such that $\X=\G^T\H$.
Moreover, $\Colspace{\X}= \Rowspace{\G}\in \Grassm{m,r}$ and $\Rowspace{\X}=\Rowspace{\H} \in \Grassm{n,r}$.

\end{lemma}

The next lemma shows a connection between the subspace distance and the rank distance and is a special case of \cite[Theorem~1]{Gadouleau2010ConstantRank}. 
It plays a non-negligible role in the proof of our bounds, therefore, we give the proof for two matrices of same rank and use the subspace distance (in \cite{Gadouleau2010ConstantRank}, the injection distance is used).
\begin{lemma}[Connection Subspace and Rank Distance, \cite{Gadouleau2010ConstantRank}]\label{lem:connection_rankdist_subspacedist}
Let $\X,\Y$ be two matrices in $\Fq^{m\times n}$ with $\rk(\X)=\rk(\Y)$. Then:
\begin{align*}
&\frac{1}{2} \; d_S \big(\Rowspace{\X},\Rowspace{\Y}\big) + \frac{1}{2}\; d_S \big(\Colspace{\X},\Colspace{\Y}\big)\\ 
&\leq d_R (\X,\Y)\\
&\leq \min \left\lbrace \frac{1}{2}\; d_S \big(\Rowspace{\X},\Rowspace{\Y}\big), \frac{1}{2}\; d_S \big(\Colspace{\X},\Colspace{\Y}\big) \right\rbrace\\
&\hspace{7ex}+\rk(\X).
\end{align*}
\end{lemma}
\begin{proof}
Let us denote $r \defeq \rk(\X)=\rk(\Y)$.
As in Lemma~\ref{lem:rank_decomp}, we decompose $\X=\C^T\R$ and $\Y=\D^T\S$, where $\C,\D \in \Fq^{r\times m}$ and $\R,\S \in \Fq^{r\times n}$ and all four matrices have full rank.
Hence, $\X-\Y= (\C^T |-\D^T)\cdot (\R^T | \S^T)^T$.
In general, it is well-known that $\rk(\A\B)\leq \min\{\rk(\A),\rk(\B)\}$ and $\rk(\A\B)\geq \rk(\A)+\rk(\B)-n$ when $\A$ has $n$ columns and $\B$ has $n$ rows. Therefore,
\begin{align}
&\rk(\C^T |-\D^T) + \rk(\R^T | \S^T) - 2r,\nonumber\\
&\leq \rk(\X-\Y) =\rk\left((\C^T |-\D^T)\cdot (\R^T | \S^T)^T\right)\label{eq:rank_est_first}\\
&\leq\min\left\{\rk(\C^T |-\D^T), \rk (\R^T | \S^T)\right\}.\nonumber
\end{align}
Let $\Colspace{\C^T}+\Colspace{\D^T}$ denote the smallest subspace containing both column spaces. Then,
\begin{align*}
\rk&(\C^T |-\D^T)\\ 
&= \dim(\Colspace{\C^T}+\Colspace{\D^T})\\
&=\dim(\Colspace{\C^T}+\Colspace{\D^T})\\
&\qquad-\frac{1}{2}\left\{\dim(\Colspace{\C^T})+\dim(\Colspace{\D^T})\right\}\\
&\qquad+\frac{1}{2}\left\{\dim(\Colspace{\C^T})+\dim(\Colspace{\D^T})\right\}\\
&=\frac{1}{2}\; d_S(\Colspace{\C^T},\Colspace{\D^T}) + r\\
&=\frac{1}{2}\; d_S(\Colspace{\X},\Colspace{\Y}) + r,
\end{align*}
and in the same way 
\begin{equation*}
\rk(\R^T | \S^T) = \frac{1}{2} \; d_S(\Rowspace{\X},\Rowspace{\Y}) + r.
\end{equation*}
Inserting this into \eqref{eq:rank_est_first}, the statement follows.
\end{proof}
Lemma~\ref{lem:connection_rankdist_subspacedist} can equivalently be derived from \cite[Equation~(4.3)]{MatsagliaStyan-EqualitiesAndInequalitiesForRanksOfMatrices}, which also results in \eqref{eq:rank_est_first} and with the same reformulations for the subspace distance, we also obtain the result.

For the proof of the upper bound in Theorem~\ref{theo:upper_bound_listsize} (see Section~\ref{sec:upper_bound}), the following upper bound on the maximum cardinality of a constant-rank code is applied. It shows a relation between the maximum cardinalities of a (not necessarily linear) constant-rank and a constant-dimension code.
\begin{proposition}[Maximum Cardinality,\cite{Gadouleau2010ConstantRank}]\label{prop:max_cardinality}
For all $q$ and $1 \leq \delta \leq r \leq n \leq m$, the \textbf{maximum} cardinality of a $\CRC{n,M,d_R = \delta+r,r}$ constant-rank code over $\Fqm$ is upper bounded by the maximum cardinality of a constant-dimension code as follows:
\begin{equation*}
\maxCardinalityRank{n,d_R = \delta+r,r} \leq \maxCardinalitySub{n,d_S = 2\delta,r}.
\end{equation*}
\end{proposition}

However, the connections between constant-dimension and constant-rank codes are even more far-reaching.
The following proposition shows explicitly how to construct constant-rank codes out of constant-dimension codes and is a generalization of 
\cite[Proposition~3]{Gadouleau2010ConstantRank} to arbitrary cardinalities (in \cite{Gadouleau2010ConstantRank} both constant-dimension codes used in the construction have the same cardinality).
\begin{proposition}[Construction of a Constant-Rank Code]\label{prop:gado_cdc_crc}
Let $\mycode{M}$ be a $\CDC{m,|\mycode{M}|,d_{S,M},r}$ and $\mycode{N}$ be a $\CDC{n,|\mycode{N}|,d_{S,N},r}$ constant-dimension code with $r \leq \min\{n,m\}$ and cardinalities $|\mycode{M}|$ and $|\mycode{N}|$.
Then, there exists a $\CRC{n,M_R,d_R,r}$ constant-rank code $\mycode{C}$ of cardinality $M_R=\min\{|\mycode{M}|,|\mycode{N}|\}$ with $\Colspace{\mycode{C}} \subseteq \mycode{M}$ and $\Rowspace{\mycode{C}} \subseteq \mycode{N}$. Furthermore, the minimum rank distance $d_R$ of $\mycode{C}$ is
\begin{equation*}
d_R \geq \frac{1}{2}\; d_{S,M} +\frac{1}{2}\; d_{S,N},
\end{equation*}
and if $|\mycode{M}|=|\mycode{N}|$ additionally:
\begin{equation*}
d_R \leq\frac{1}{2} \min\{d_{S,M},d_{S,N}\} + r.
\end{equation*}
\end{proposition}
\begin{proof}
Let $\G_i \in \Fq^{r \times m}$ and $\H_i\in \Fq^{r \times n}$ for $i=1,\dots,\min\{|\mycode{M}|,|\mycode{N}|\}$ be full-rank matrices, whose row spaces are $\min\{|\mycode{M}|,|\mycode{N}| \}$
codewords (which are subspaces themselves) of $\mycode{M}$ and $\mycode{N}$, respectively.

Let $\mycode{C}$ be a $\CRC{n,M_R,d_R,r_R}$ constant-rank code, defined by the set of codewords $\A_i = \G_i^T \H_i$ for $i=1,\dots,M_R$, where $M_R=\min\{|\mycode{M}|,|\mycode{N}| \}$.
All such codewords $\A_i$ are distinct, since the row spaces of all $\G_i$, respectively $\H_i$, are different.
These codewords $\A_i$ are $m\times n$ matrices of rank exactly $r_R = r$ since $\G_i \in \Fq^{r \times m}$ and $\H_i\in \Fq^{r \times n}$ have rank $r$.
The cardinality is $|\mycode{C}|=\min\{ |\mycode{M}|,|\mycode{N}|\}$ and $\Colspace{\mycode{C}} \subseteq \mycode{M}$ and $\Rowspace{\mycode{C}} \subseteq \mycode{N}$ by Lemma~\ref{lem:rank_decomp}.

The lower bound on the minimum rank distance follows with Lemma~\ref{lem:connection_rankdist_subspacedist} for two different $\A_i,\A_j$:
\begin{align*}
d_R &\geq \frac{1}{2} \; d_S (\Rowspace{\A_i},\Rowspace{\A_j}) + \frac{1}{2}\; d_S (\Colspace{\A_i},\Colspace{\A_j})\\
&\geq\frac{1}{2}\; d_{S,N} + \frac{1}{2}\; d_{S,M}.
\end{align*}
If $|\mycode{M}|=|\mycode{N}|$, there exist two matrices $\A_i,\A_j$ such that $d_S (\Rowspace{\A_i},\Rowspace{\A_j})= d_{S,N}$. Then,  Lemma~\ref{lem:connection_rankdist_subspacedist} gives $d_R \leq d_{S,N} +r $. If we choose $\A_i$ and $\A_j$ such that $d_S (\Colspace{\A_i},\Colspace{\A_j}) = d_{S,M}$, then $d_R\leq d_{S,M} +r $ and the statement follows.
\end{proof}

\subsection{Constant-Dimension Codes from Lifted MRD Codes}
%

The maximum cardinality of constant-dimension codes and explicit constructions of codes with high cardinality have been investigated in several papers, see \cite{Wang2003Linear,koetter_kschischang,silva_rank_metric_approach,Xia2009Johnson,Etzion2009ErrorCorrecting,Skachek2010Recursive,SilbersteinEtzion-EnumerativeCodingForGrassmannianSpace-2011,Etzion2011ErrorCorrecting,Bachoc2012Bounds,EtzionSilberstein-CodesDesignsRelLiftedMRD-2012}.
However, for our application, the construction from \cite{silva_rank_metric_approach} based on lifted MRD codes (e.g. Gabidulin codes) is sufficient. 
These constant-dimension codes are shown for some explicit parameters in Lemma~\ref{lem:subspacecode_lifted_mrd} and Corollary~\ref{cor:subspacecode_lifted_mrd_dodd}, where the \emph{lifting} is defined as follows.

\begin{definition}[Lifting of a Matrix or a Code, \cite{silva_rank_metric_approach}]
Consider the mapping
\begin{align*}
\Fq^{r \times (n-r)} &\rightarrow \Grassm{n,r}\\
\X &\mapsto\myspace{I}(\X)= \Rowspace{[\I_r \ \X]},
\end{align*}
where $\I_r$ denotes the $r \times r $ identity matrix.
The subspace $\myspace{I}(\X)$ is called \textbf{lifting} of the matrix $\X$. If we apply this map on all code matrices of a block code $\mycode{C}$, then the constant-dimension code $\myspace{I}(\mycode{C})$ is called lifting of $\mycode{C}$.
\end{definition}

\begin{lemma}[Lifted MRD Code, \cite{silva_rank_metric_approach}]\label{lem:subspacecode_lifted_mrd}
Let $d$ be an even integer, let $n/2 \geq \tau \geq d/2$. 
Let a linear $\MRDlinear{\tau,\tau-d/2+1}$ code $\mycode{C}$ over $\mathbb{F}_{q^{n-\tau}}$ of length $\tau$, minimum rank distance $d_R = d/2$ and cardinality $M_R$ be given.
  
Then, the lifting of the transposed codewords, i.e.,
\begin{equation*}
\myspace{I}(\mycode{C}^T)\defeq\big\{\myspace{I}(\C^T) = \Rowspace{[\I_{\tau} \ \C^T]} : \C \in \mycode{C}\big\}
\end{equation*}
is a $\CDC{n,M_S,d_S,\tau}$ constant-dimension  code of cardinality $M_S=M_R=q^{(n-\tau)(\tau-d/2+1)}$, minimum subspace distance $d_S=d$ and lies in the Grassmannian $\Grassm{n,\tau}$.
\end{lemma}
\begin{proof} 
Let $\C_i \in \Fq^{(n-\tau)\times \tau}$, for $i=1,\dots, M_R$, denote the codewords of \mycode{C} in matrix representation.
The dimension of each subspace $\myspace{I}(\C_i^T)$ is $\tau$ since $\rk([\I_{\tau} \ \C_i^T]) = \tau$ for all $i=1,\dots,M_R$. The cardinality of this constant-dimension code is the same as the cardinality of the MRD code, which is $M_R=q^{(n-\tau)(\tau-d/2+1)}$. 
The subspace distance of the constant-dimension code is two times the rank distance of the MRD code (see \cite[Proposition~4]{silva_rank_metric_approach}). 
The restriction $\tau \leq n-\tau$ has to hold since the length of the MRD code has to be at most the extension degree of the finite field.
\end{proof}

Corollary~\ref{cor:subspacecode_lifted_mrd_dodd} shows how constant-dimension codes can be constructed from MRD codes when $d$ is odd. 
\begin{corollary}[Lifted MRD Codes, Case 2]\label{cor:subspacecode_lifted_mrd_dodd}
Let $d$ be an odd integer and let $\tau \geq \nicefrac{(d-1)}{2} +1$.  
Then, 
\begin{itemize}
\item for $\tau \leq m-\tau$ and for a linear $\MRDlinear{\tau,\tau-\nicefrac{(d-1)}{2}+1}$ code $\mycode{C}$ over $\mathbb{F}_{q^{m-\tau}}$, the lifting $\myspace{I}(\mycode{C}^T)$ is a $\CDC{m,M_S,d_S = d-1,\tau}$ constant-dimension code  of cardinality $M_S = q^{(m-\tau)(\tau-\nicefrac{(d-1)}{2}+1)}$,
\item for $\tau \leq n-\tau$, $n \leq m$ and for a linear $\MRDlinear{\tau,\tau-\nicefrac{(d+1)}{2}+1}$ code $\mycode{C}$ over $\mathbb{F}_{q^{n-\tau}}$, the lifting $\myspace{I}(\mycode{C}^T)$ is a $\CDC{n,M_S,d_S = d+1,\tau}$ constant-dimension code of cardinality $M_S=q^{(n-\tau)(\tau-\nicefrac{(d+1)}{2}+1)}=q^{(n-\tau)(\tau-\nicefrac{(d-1)}{2})} <q^{(m-\tau)(\tau-\nicefrac{(d-1)}{2}+1)} $.
\end{itemize}
\end{corollary}

Lifted MRD codes are said to be \emph{asymptotically} optimal constant-dimension codes since 
the ratio of their cardinality to the upper bounds is a constant \cite{silva_rank_metric_approach}. There are constant-dimension codes of higher cardinality, e.g. the construction from \cite{Etzion2009ErrorCorrecting}. However, for our approach, lifted MRD codes are sufficient, since scalar factors do not change the asymptotic behavior and since such constant-dimension codes exist for any $\tau$ and $d$ when $\tau \leq n-\tau$.

\subsection{Problem Statement}
We analyze the question of \emph{polynomial-time list decodability} of rank-metric codes. 
Thus, we want to bound the maximum number of codewords in a ball of radius $\tau$ around a received word $\r$. This number is called the maximum \emph{list size} $\ell$ in the following. 
The worst-case complexity of a possible list decoding algorithm directly depends on $\ell$.


\begin{problem}[Maximum List Size]
Let $\mycode{C}$ be an $\code{n,M,d}$ code over $\Fqm$ of length $n \leq m$, cardinality $M$ and minimum rank distance $d_R = d$. Let $\tau < d$. 
Find lower and upper bounds on the maximum number of codewords $\ell$ in a ball of 
rank radius $\tau$ around a word $\r= (r_0 \ r_1 \ \dots \ r_{n-1}) \in \Fqm^n$. Hence, find a bound on
\begin{equation*}
\ell \defeq \ell\big(m,n,d,\tau\big) \defeq \max_{\r \in \Fqm^n}\Big\{\big|\mycode{C} \cap \Ball{\tau}{\r}\big|\Big\}.
\end{equation*}
\end{problem}
When the paramters $m,n,d,\tau$ are clear from the context, we use the short-hand notation $\ell$ for the maximum list size.
For an upper bound on $\ell$, we have to show that the bound holds for \emph{any} received word $\r$, 
whereas for a lower bound on $\ell$ it is sufficient to show that there exists (at least) one $\r$ for which 
this bound on the list size is valid.

Moreover, if we restrict $\mycode{C}$ to be a Gabidulin code rather than an arbitrary rank-metric code, the task becomes more difficult due to the additional imposed structure of the code.

Let us denote the list of all codewords of an $\code{n,M,d}$ code $\mycode{C}$ in the ball of rank radius $\tau$ around a given word $\r$ by:
\begin{align}
\List\big(\mycode{C},&\r\big) \defeq \mycode{C} \cap \Ball{\tau}{\r}\label{eq:def_listsize} \\
&= \big\{ \c_1,\c_2,\dots,\c_{|\List|} : \c_i \in \mycode{C} \; \text{and}\; \rk(\r-\c_i) \leq \tau, \; \forall i \big\}. \nonumber
\end{align}
Clearly, the cardinality is $|\List| \defeq |\List\big(\mycode{C},\r\big)| \leq \ell$. 


\section{A Lower Bound on the List Size of\\ Gabidulin Codes}\label{sec:bound_gabidulin}
In this section, we provide a lower bound on the list size when decoding Gabidulin codes. 
The proof is
based on the evaluation of linearized polynomials and is inspired by Justesen and H\o{}holdt's \cite{Justesen2001Bounds} and 
Ben-Sasson, Kopparty, and Radhakrishna's \cite{BenSasson2010Subspace} approaches for bounding the list size of Reed--Solomon codes. 
\begin{theorem}[Bound I: Lower Bound on the List Size]\label{theo:lower_bound_gabidulin}
Let the linear Gabidulin code $\Gab{n,k}$ over $\Fqm$ with $n \leq m$ and $d_R=d=n-k+1$ be given. Let $\tau < d$. 
Then, there exists a word $\r \in \Fqm^n$ such that the maximum list size $\ell$ satisfies
\begin{align}
 \ell = \ell\big(m,n,d,\tau\big) &\geq\big|\Gab{n,k} \cap \Sphere{\tau}{\r}\big|\geq \frac{\quadbinom{n}{n-\tau}}{(q^m)^{n-\tau-k}} \nonumber\\
&\geq q^m q^{\tau(m+n) -\tau^2-md},\label{eq:listsize_polys}
 \end{align}
 and for the special case of $n=m$: 
 \begin{equation*}
 \ell \geq q^n q^{2n\tau - \tau^2 - nd}.
 \end{equation*}
\end{theorem}

 \begin{proof}
 Since we assume $\tau < d = n-k+1$, also $ k-1 < n-\tau $ holds. 
Let us consider all monic linearized polynomials of $q$-degree exactly $n-\tau$ whose root spaces have dimension 
 $n-\tau$ and all roots lie in $\Fqn$. There are exactly (see e.g. \cite[Theorem 11.52]{Berlekamp1984Algebraic})
 $\quadbinom{n}{n-\tau}$
such polynomials. 

Now, let us consider a subset of these polynomials, denoted by $\mathcal P$: all polynomials where
the $q$-monomials of
$q$-degree greater than or equal to $k$ have the same coefficients. 
Due to the pigeonhole principle, there exist coefficients such that the number of 
 such polynomials is 
 \begin{equation*}
|\mathcal P| \geq\frac{\quadbinom{n}{n-\tau}}{(q^m)^{n-\tau-k}},
 \end{equation*}
since there are $(q^m)^{n-\tau-k}$ possibilities to choose the 
  highest $n-\tau -(k-1)$ coefficients of a \emph{monic} linearized polynomial with coefficients $\Fqm$. 
 
Note that 
the difference of any two polynomials in $\mathcal P$ is a linearized polynomial of $q$-degree
strictly less than $k$ and therefore the evaluation polynomial of a codeword of $\Gab{n,k}$. 

Let $\r$ be the evaluation of $p(x) \in \mathcal P$ at a basis $\mathcal A = \{\alpha_0,\alpha_1,\dots,\alpha_{n-1}\}$ of $\Fqn$ over $\Fq$:
 \begin{equation*}
 \r = (r_0 \ r_1 \ \dots \ r_{n-1}) = (p(\alpha_0) \ p(\alpha_1) \ \dots \ p(\alpha_{n-1})).
 \end{equation*}
Further, let also $q(x) \in \mathcal P$, then 
 $p(x)-q(x)$ has $q$-degree less than $k$. Let
 $\c$ denote the evaluation of $p(x)-q(x)$ at $\mathcal A$.
Then, $\r-\c$ is the evaluation of $p(x)-p(x)+q(x) = q(x) \in \mathcal P$, whose 
root space has dimension $n-\tau$ and all roots lie in $\Fqn$.  
 Thus, $\dim \ker(\r-\c) = n-\tau$ and $\dim \Colspace{\r-\c} = \rk(\r-\c) = \tau$. 
 
 Therefore, for \emph{any} $q(x) \in \mathcal P$, the evaluation of $p(x)-q(x)$
 is a codeword of $\Gab{n,k}$ and has rank distance $\tau$ from $\r$. 
 Hence,
 \begin{equation*}
\big|\Gab{n,k} \cap  \Sphere{\tau}{\r}\big| \geq |\mathcal{P}|.
 \end{equation*}
Using \eqref{eq:bounds_gaussian_binomial}, this provides the following lower bound on the maximum list size:
 \begin{equation*}
 \ell \geq |\mathcal{P}|
 \geq \frac{q^{(n-\tau)\tau}}{(q^m)^{n-\tau-k}} 
\geq q^m q^{\tau(m+n) -\tau^2-md},
 \end{equation*}
 and for $n=m$ the special case follows.
 \end{proof}
 This lower bound is valid for any $\tau < d$, but we want to know, which is the smallest value for $\tau$ such that this expression grows \emph{exponentially} in $n$.
 
  For arbitrary $n\leq m$, we can rewrite \eqref{eq:listsize_polys} by
  \begin{equation*}
  \ell \geq q^{m(1-\epsilon)} \cdot q^{\tau(m+n) -\tau^2-m(d-\epsilon)},
  \end{equation*}
  where the first part is exponential in $n \leq m$ for any $0 \leq \epsilon < 1$. 
  The second exponent is positive for 
   \begin{equation*}
  \tau \geq \frac{m+n}{2}- \sqrt{\frac{(m+n)^2}{4}-m(d-\epsilon)}\defeq \tau_{J}^*.
   \end{equation*}
 For $n=m$, 
this simplifies to
  \begin{equation*}\label{eq:tau_lessjohnson}
 \tau \geq n- \sqrt{n(n-d+\epsilon)} \defeq \tau_{J}.
  \end{equation*}
Therefore, our lower bound \eqref{eq:listsize_polys} 
shows that the maximum list size is exponential in $n$ for any $\tau \geq \tau_{J}^*$.
For $n=m$, the value $\tau_J$ is basically the Johnson radius 
for codes in Hamming metric. 
Notice that Faure obtained a similar result in \cite{Faure2006Average} by using probabilistic arguments.

This reveals a difference between the known limits to list decoding of Gabidulin and Reed--Solomon codes. 
For Reed--Solomon codes, polynomial-time list decoding up to the Johnson radius can be accomplished 
by the Guruswami--Sudan algorithm. 
However, it is not proven that the Johnson radius is tight for Reed--Solomon codes, i.e., it is not known if the list size is polynomial in $n$ between the Johnson radius and the known exponential lower bounds (see e.g. \cite{Justesen2001Bounds,BenSasson2010Subspace}).

\begin{remark}[Alternative Proof]

The result of Theorem~\ref{theo:lower_bound_gabidulin} can also be obtained by interpreting the decoding list as a constant-rank code as in Subsection~\ref{sec:connection_crc_list}. 
For this purpose, we can use \cite[Lemma~2]{Gadouleau2010ConstantRank} as follows.

Let $\mycode{C}$ be a $\Gab{n,n-d+1}$ of minimum rank distance $d$ and $\mycode{B}$ be a $\Gab{n,d-\tau}$ code of minimum rank distance $n-d+\tau+1$.
Let $\mycode{C}$ be defined as in Definition~\ref{def:gabidulin_code} with the elements $\alpha_0,\alpha_1,\dots,\alpha_{n-1} \in \Fqm$, which are linearly independent over $\Fq$, and let
$\mycode{B}$ be defined with $\alpha_0^{[n-d+1]}, \alpha_1^{[n-d+1]}, \dots, \alpha_{n-1}^{[n-d+1]}$. The corresponding generator matrices are denoted by $\G_{\mycode{C}}$ and $\G_{\mycode{B}}$.

Then, the direct sum code $\mycode{C} \oplus \mycode{B}$ has the generator matrix $(\G_{\mycode{C}}^T \ \G_{\mycode{B}}^T)^T$ and is a $\Gab{n,n-\tau+1}$ code with minimum rank distance $\tau$.

The rank weight distribution of MRD codes can be found in \cite[Section~3]{Gabidulin_TheoryOfCodes_1985} and therefore the number of codewords of rank $\tau$ in $\mycode{C} \oplus \mycode{B}$ is
\begin{equation*}
W_{\tau}(\mycode{C} \oplus \mycode{B}) = \quadbinom{n}{\tau} (q^m-1).
\end{equation*}
The cardinality of the code $\mycode{B}$ is $|\mycode{B}| = q^{m(d-\tau)}$ and therefore, with the pigeonhole principle, there exists a vector $\b \in \mycode{B}$ such that the number of codewords of rank $\tau$ in the translated code $\mycode{C} \oplus \b$ is lower bounded by
\begin{equation}\label{eq:lowerboundgabidulin_alternativ}
W_{\tau}(\mycode{C} \oplus \b) \geq \frac{\quadbinom{n}{\tau} (q^m-1)}{q^{m(d-\tau)}}.
\end{equation}
This means the number of codewords of $\mycode{C}$ in rank distance $\tau$ from $\b$ is $W_{\tau}(\mycode{C} \oplus \b)$ and \eqref{eq:lowerboundgabidulin_alternativ} yields the same lower bound on $\ell\big(m,n,d,\tau\big)$ as Theorem~\ref{theo:lower_bound_gabidulin}.
\end{remark}


\section{Bounds on the List Size of Arbitrary Rank-Metric Codes}\label{sec:bounds}
\subsection{Connection between Constant-Rank Codes and the List Size}\label{sec:connection_crc_list}
Before proving our bounds, let us explain the connection between the list size for decoding a certain rank-metric code and the cardinality of a certain constant-rank code. 
As in \eqref{eq:def_listsize}, denote the list of codewords for a decoding radius $\tau <d $ and an $\code{n,M,d_R=d}$ code $\mycode{C}$ by 
\begin{align*}
\List\big(\mycode{C},\r\big) &= \big\{\c_1,\c_2,\dots,\c_{|\List|}\big\} =  \mycode{C} \cap \Ball{\tau}{\r}
= \sum_{i=0}^{\tau}\big(\mycode{C}\cap\Sphere{i}{\r} \big),
\end{align*}
for some (received) word $\r \in \Fqm^n$. 
If we consider only the codewords with rank distance \emph{exactly} $\tau$ from the received word, i.e., on the sphere $\Sphere{\tau}{\r}$:
\begin{equation*}
\big\{\c_1,\c_2,\dots,\c_{\overline{\ell}}\big\} \defeq \mycode{C}\cap \Sphere{\tau}{\r},
\end{equation*} 
we obtain a lower bound on the maximum list size: $\ell \geq\overline{\ell} = |\mycode{C} \cap \Sphere{\tau}{\r}|$.

Now, consider a \emph{translate} of all codewords on the sphere of radius $\tau$ as follows:
\begin{equation*}
\overline{\mathcal L}\big(\mycode{C},\r\big) \defeq \big\{\r-\c_1,\r-\c_2,\dots,\r-\c_{\overline{\ell}}\big\}.
\end{equation*}
This set $\overline{\List}\big(\mycode{C},\r\big)$ is a $\CRC{n,M_R,d_R\geq d,\tau}$ constant-rank code over $\Fqm$ since $\rk(\r-\c_i)=\tau$ for all $i=1,\dots,\overline{\ell}$ and
its minimum rank distance is at least $d$, since 
\begin{equation*}
\rk(\r-\c_i-\r+\c_j) = \rk(\c_i-\c_j) \geq d, \quad \forall i,j,\; i\neq j. 
\end{equation*}
The cardinality of this constant-rank code is exactly $M_R=\overline{\ell}$.
For $\tau < d$, this constant-rank code is non-linear (or a translate of a linear code if $\mycode{C}$ is linear), since the rank of its codewords is $\tau$, but its minimum distance is at least $d$.

Hence, a translate of the list of all codewords of rank distance exactly $\tau$ from the received word can be interpreted as a constant-rank code.
This interpretation makes it possible to use bounds on the cardinality of a constant-rank codes to obtain bounds on the list size $\ell$ for decoding rank-metric codes.

\subsection{An Upper Bound on the List Size}\label{sec:upper_bound}
In this subsection, we derive an upper bound on the list size when decoding rank-metric codes.
This upper bound holds for \emph{any} rank-metric code and \emph{any} received word. 

\begin{theorem}[Bound II: Upper Bound on the List Size]\label{theo:upper_bound_listsize}
Let $\dhalfnicefrac \leq \tau < d\leq n\leq m$. 
Then, for \textbf{any} $\code{n,M, d}$ code $\mycode{C}$ in rank metric, 
the maximum list size is upper bounded as follows:
 \begin{align}
  \ell&=\ell\big(m,n,d,\tau\big) =\max_{\r \in \Fqm^n}\Big\{\big| \mycode{C}\cap\Ball{\tau}{\r} \big|\Big\}\nonumber\\
  &\leq 1+\sum\limits_{t=\dhalffrac+1}^{\tau} \frac{\quadbinom{n}{2 t + 1-d}}{\quadbinom{t}{2 t + 1-d}}\nonumber\\
  &\leq 1+  4\sum\limits_{t=\dhalffrac+1}^{\tau} q^{(2t-d+1)(n-t)}\nonumber\\
 &\leq 1 + 4\cdot \big(\tau-\left\lfloor\tfrac{d-1}{2}\right\rfloor\big)\cdot q^{(2\tau-d+1)(n-\dhalfnicefrac-1)}.\label{eq:upper_bound_rank}
  \end{align}
\end{theorem}

\begin{proof}
Let $\{\c_1,\c_2,\dots,\c_{\overline{\ell}}\}$ denote
the intersection of the sphere $\Sphere{t}{\r}$ in rank metric around $\r$ and the code $\mycode{C}$. 
As explained in Section~\ref{sec:connection_crc_list},
\begin{equation*}
\overline{\List}\big(\mycode{C},\r\big) = \big\{\r-\c_1,\r-\c_2,\dots,\r-\c_{\overline{\ell}}\big\}
\end{equation*}
can be seen as a $\CRC{n,M_R,d_R \geq d,t}$ constant-rank code over $\Fqm$ for any word $\r \in \Fqm^n$. 
Therefore, for any word $\r \in \Fqm^n$, the cardinality of $\overline{\List}\big(\mycode{C},\r\big)$ can be upper bounded by the maximum cardinality of a constant-rank code with the corresponding parameters:
\begin{align*}
|\overline{\List}\big(\mycode{C},\r\big)| =\big|\mycode{C}\cap \Sphere{t}{\r}\big| &\leq \maxCardinalityRank{n,d_R \geq d,t}\\ &\leq \maxCardinalityRank{n,d,t}.
\end{align*}
We can upper bound this maximum cardinality by Proposition~\ref{prop:max_cardinality} with $\delta = d-t$ and $r=t$ by the maximum cardinality of a constant-dimension code:
\begin{equation*}
\maxCardinalityRank{n,d,t} \leq \maxCardinalitySub{n,d_S = 2(d-t),t}.
\end{equation*}
For upper bounding the cardinality of such a constant-dimension code, we use the Wang--Xing--Safavi-Naini bound \cite{Wang2003Linear} (often also called \emph{anticode bound}) and obtain:
\begin{equation}\label{eq:card_cdc_wangxing}
\maxCardinalitySub{n,d_S = 2(d-t),t} \leq \frac{\quadbinom{n}{t-(d-t)+1}}{\quadbinom{t}{t-(d-t)+1}}.
\end{equation}
In the ball of radius $\dhalfnicefrac$ around $\r$, there can be at most one codeword of $\mycode{C}$ and therefore, the contribution to the list size is at most one.
For higher $t$, we sum up \eqref{eq:card_cdc_wangxing} from $t= \dhalfnicefrac+1$ up to $\tau$, use the upper bound on the Gaussian binomial \eqref{eq:bounds_gaussian_binomial} and upper bound the sum.
\end{proof}

This upper bound gives (almost) the same upper bound as we showed in \cite[Theorem~2]{Wachterzeh-BoundsListDecodingGabidulin_conf} and it can slightly be improved if we use better upper bounds on the maximum cardinality of constant-dimension codes instead of \eqref{eq:card_cdc_wangxing} in the derivation, for example the iterated Johnson bound for constant-dimension codes \cite[Corollary~3]{Xia2009Johnson}. 
In this case, we obtain:
\begin{align*}
&\ell=\ell\big(m,n,d,\tau\big) \leq 1+\\
&\sum\limits_{t=\dhalffrac+1}^{\tau}
\left\lfloor \frac{q^n-1}{q^{t}-1} \left \lfloor \frac{q^{n-1}-1}{q^{t-1}-1} \left\lfloor \dots \left \lfloor \frac{q^{n+d-2t}-1}{q^{d-t}-1} \right \rfloor  \dots \right \rfloor\right \rfloor\right \rfloor.
\end{align*}
However, the Wang--Xing--Safavi-Naini bound provides a nice closed-form expression and is asymptotically tight. Therefore, using better upper bounds for constant-dimension codes does not change the asymptotic behavior of our upper bound. 

Unfortunately, our upper bound on the list size of rank-metric codes is exponential in the length of the code for any $\tau > \dhalfnicefrac$ and not polynomial as the Johnson bound for Hamming metric.
However, the lower bound of Section~\ref{sec:lower_bound} shows that any upper bound depending only on the length $n \leq m$ and the minimum rank distance $d$ has to be exponential in $(\tau-\dhalfnicefrac)(n-\tau)$, since there exists a rank-metric code with such a list size.

\subsection{A Lower Bound on the List Size}\label{sec:lower_bound}

In this subsection, we prove the most significant difference to codes in Hamming metric. 
We show the existence of a rank-metric code with exponential list size for any decoding radius \emph{greater than half the minimum distance}.

First, we prove the existence of a certain constant-rank code in the following theorem.

\begin{theorem}[Constant-Rank Code]\label{theo:exists_crc}
Let $\dhalfnicefrac+1 \leq \tau < d\leq n\leq m$ and $\tau \leq n-\tau$. 
Then, there exists a $\CRC{n,M_R,d_R\geq d,\tau}$ constant-rank code over $\Fqm$ of cardinality $M_R=q^{(n-\tau)(\tau-\dhalfnicefrac)}$.
\end{theorem}
\begin{IEEEproof}
First, assume $d$ is even. 
Let us construct a $\CDC{m,|\mycode{M}|,d,\tau}$ constant-dimension code $\mycode{M}$ and a $\CDC{n,|\mycode{N}|,d,\tau}$ code $\mycode{N}$ by lifting an $\MRDlinear{\tau,\tau-\nicefrac{d}{2}+1}$ code over $\mathbb{F}_{q^{m-\tau}}$ of minimum rank distance $\nicefrac{d}{2}$ and an $\MRDlinear{\tau,\tau-\nicefrac{d}{2}+1}$ code over $\mathbb{F}_{q^{n-\tau}}$ of minimum rank distance $\nicefrac{d}{2}$ as in Lemma~\ref{lem:subspacecode_lifted_mrd}. 
Then, with Lemma~\ref{lem:subspacecode_lifted_mrd}:
\begin{equation*}
|\mycode{N}|=q^{(n-\tau)(\tau-\nicefrac{d}{2}+1)} \leq |\mycode{M}|=q^{(m-\tau)(\tau-\nicefrac{d}{2}+1)}.
\end{equation*}
From Proposition~\ref{prop:gado_cdc_crc}, we know therefore
there exists a $\CRC{n,M_R,d_R,\tau}$ code of cardinality 
\begin{equation*}
M_R=\min\{|\mycode{N}| ,|\mycode{M}| \}=q^{(n-\tau)(\tau-\nicefrac{d}{2}+1)}=q^{(n-\tau)(\tau-\dhalfnicefrac)}.
\end{equation*}
For its rank distance, the following holds with Proposition~\ref{prop:gado_cdc_crc}:
\begin{equation*}
d_R \geq \frac{1}{2}\; d_{S,M} +\frac{1}{2}\; d_{S,N} = d.
\end{equation*}

Second, assume $d$ is odd. 
Let $\mycode{M}$ be a $\CDC{m,|\mycode{M}|,d-1,\tau}$ code and $\mycode{N}$ be a $\CDC{n,|\mycode{N}|,d+1,\tau}$ code, constructed as in Corollary~\ref{cor:subspacecode_lifted_mrd_dodd}. 
Then, 
\begin{equation*}
|\mycode{N}|=q^{(n-\tau)(\tau-\nicefrac{(d+1)}{2}+1)} \leq|\mycode{M}| =q^{(m-\tau)(\tau-\nicefrac{(d-1)}{2}+1)}.
\end{equation*}
From Proposition~\ref{prop:gado_cdc_crc}, we know that
there exists a $\CRC{n,M_R,d_R,\tau}$ code of cardinality 
\begin{align*}
M_R=\min\{|\mycode{N}| ,|\mycode{M}| \}=|\mycode{N}|
&=q^{(n-\tau)(\tau-\nicefrac{(d-1)}{2})}\\
&=q^{(n-\tau)(\tau-\dhalfnicefrac)}.
\end{align*}
With Proposition~\ref{prop:gado_cdc_crc}, the rank distance $d_R$ is lower bounded by:
\begin{equation*}
d_R \geq \frac{1}{2}\; d_{S,M} +\frac{1}{2}\; d_{S,N} = \frac{1}{2}\;(d-1) +\frac{1}{2}\;(d+1) =d.
\end{equation*}
\end{IEEEproof}

This constant-rank code can now directly be used to show the existence of a rank-metric code with exponential list size.

\begin{theorem}[Bound III: Lower Bound on the List Size]\label{theo:lower_bound_list_size}
Let $\dhalfnicefrac+1 \leq \tau < d\leq n$ and $\tau \leq n-\tau$. 
Then, there exists an $\code{n,M,d_R \geq d}$ code $\mycode{C}$ over $\Fqm$ of length $n\leq m$ and minimum rank distance $d_R \geq d$, and a word $\r \in \Fqm^n$
such that 
\begin{equation}\label{eq:lower_bound_rankmetric}
\ell=\ell\big(m,n,d,\tau\big) 
\geq \big|\mycode{C}\cap \Ball{\tau}{\r}\big| \geq q^{(n-\tau)(\tau-\dhalfnicefrac)}.
\end{equation}
\end{theorem}
\begin{IEEEproof}
Let the $\CRC{n,M_R,d_R \geq d,\tau}$ constant-rank code from Theorem~\ref{theo:exists_crc} consist of the codewords: 
\begin{equation*}
\big\{\a_1, \a_2,\dots, \a_{|\mycode{N}|}\big\}.
\end{equation*}
This code has cardinality $M_R=|\mycode{N}|= q^{(n-\tau)(\tau-\dhalfnicefrac)} $ (see Theorem~\ref{theo:exists_crc}).
Choose $\r = \0$, and hence, $\rk(\r-\a_i)=\rk(\a_i)=\tau$  for all $i=1,\dots,|\mycode{N}|$ since the $\a_i$ are codewords of a constant-rank code of rank $\tau$.
Moreover, 
$d_R(\a_i,\a_j)=\rk(\a_i-\a_j)  \geq d$ since the constant-rank code has minimum rank distance at least $ d$.

Therefore, $\a_1, \dots, \a_{|\mycode{N}|}$ are codewords of an $\code{n,M,d_R\geq d}$ code $\mycode{C}$ over $\Fqm$ in rank metric, which all lie on the sphere of rank radius $\tau$ around $\r=\0$ (which is not a codeword of $\mycode{C}$).

Hence, there exists an $\code{n,M,d_R\geq d}$ code $\mycode{C}$ over $\Fqm$ of length $n \leq m$ such that $\ell \geq |\mycode{C}\cap \Ball{\tau}{\r}| \geq |\mycode{C}\cap \Sphere{\tau}{\r}| =  |\mycode{N}|=q^{(n-\tau)(\tau-\dhalfnicefrac)}$.
\end{IEEEproof}

Notice that this $\code{n,M,d_R\geq d}$ code in rank metric is non-linear since it has codewords of weight $\tau < d$, but minimum rank distance at least $d$.

For constant code rate $R = \nicefrac{k}{n}$ and constant relative decoding radius $\nicefrac{\tau}{n}$, where $\tau > \dhalfnicefrac$, \eqref{eq:lower_bound_rankmetric} gives
\begin{align*}
\ell&\geq 
q^{n^2\left(1-\nicefrac{\tau}{n})(\nicefrac{\tau}{n}-\nicefrac{1}{2}(1-R)\right)} = q^{n^2\cdot const}.
\end{align*}
Therefore, the lower bound for this $\code{n,M,d_R\geq d}$ code is exponential in $n \leq m$ for any $\tau > \dhalfnicefrac$.
Hence, Theorem~\ref{theo:lower_bound_list_size} shows that there exist rank-metric codes, where the number of codewords in a rank metric ball around the all-zero word is exponential in $n$, thereby prohibiting a polynomial-time list decoding algorithm. However, this does not mean that this holds for \emph{any} rank-metric code. In particular, the theorem does not provide a conclusion if there exists a \emph{linear} code or even a \emph{Gabidulin} code with this list size. 

\begin{remark}[Non-Zero Received Word]
The rank-metric code $\mycode{C}$ shown in Theorem~\ref{theo:lower_bound_list_size} is clearly not linear.
Instead of choosing $\r = \0$, we can choose for example $\r=\a_1$. 
The codewords of the $\CRC{n,M_R,d_R \geq d,\tau}$ constant-rank code from Theorem~\ref{theo:exists_crc} of cardinality $M_R=|\mycode{N}|= q^{(n-\tau)(\tau-\dhalfnicefrac)} $ are denoted by:
\begin{equation*}
\big\{\a_1, \a_2,\dots, \a_{|\mycode{N}|}\big\}.
\end{equation*}
Then, the following set of words
\begin{align*}
\big\{\c_1, \c_2,\dots, \c_{|\mycode{N}|}\big\} &\defeq\big\{\a_1, \a_2,\dots, \a_{|\mycode{N}|}\big\}-\a_1 \\
&\defeq \big\{\0, \a_1-\a_2,\a_1-\a_3,\dots, \a_1-\a_{|\mycode{N}|}\big\}
\end{align*}
consists of codewords of an $\code{n,M,d_R\geq d}$ code $\mycode{C}$ over $\Fqm$ since $d_R(\c_i,\c_j)=\rk(\c_i-\c_j)=\rk(\a_1-\a_i-\a_1+\a_j)  =\rk(\a_j-\a_i)  \geq d$ for $i \neq j$ since 
$\a_i, \a_j$ are codewords of the constant-rank code of minimum rank distance $d_R$.
Moreover, all codewords $\c_i$ have rank distance exactly $\tau$ from $\r$ since $\rk(\r-\c_i) = \rk(\a_i) = \tau$ and the same bound on the list size of $\mycode{C}$
follows as in Theorem~\ref{theo:lower_bound_list_size}.
This $\code{n,M,d_R\geq d}$ rank-metric code over $\Fqm$ is not necessarily linear, but also not necessarily not linear.
\end{remark}

The next corollary shows that the restriction $\tau \leq n-\tau$ does not limit the code rate for which Theorem~\ref{theo:lower_bound_list_size} shows an exponential behavior of the list size.
For the special case of $\tau=\dhalfnicefrac+1$, the condition $\tau \leq n-\tau$ is always fulfilled for even minimum distance since $d \leq n$. For odd minimum $d-1\leq n$ has to hold.
Notice that $d=n$ is a trivial code. 

\begin{corollary}[Special Case ${\tau = \dhalfnicefrac +1}$]\label{cor:special_case_plus_one}
Let $n\leq m$, $\tau= \dhalfnicefrac +1$ and $d \leq n-1$ when $d$ is odd. Then, there exists an $\code{n,M,d_R\geq   d}$ code $\mycode{C}$
and a word $\r \in \Fqm^n$ such that $|\mycode{C}\cap \Ball{\tau}{\r}| \geq q^{(n-\tau)}$.
\end{corollary}
This corollary hence shows that for any $n \leq m $ and \emph{any code rate} there exists a rank-metric code of rank distance at least $d$ whose list size can be exponential in $n$.

For the special case when $d$ is even, $\tau=d/2$ and $n=m$, the minimum rank distance of $\mycode{C}$ is \emph{exactly} $d$ since the lower and upper bound on $d_R$ in Proposition~\ref{prop:gado_cdc_crc} coincide.
\begin{corollary}[Special Case $\tau = d/2$]\label{cor:special_case_dhalf}
Let $n=m$, $d$ be even and $\tau=d/2$. Then, there exists an $\code{n,M,d_R =d}$ code $\mycode{C}$ in rank metric
and a word $\r \in \Fqm^n$ such that $|\mycode{C}\cap \Ball{\tau}{\r}| \geq q^{(n-\tau)}$.
\end{corollary}

Corollary~\ref{cor:special_case_plus_one} shows that the condition $\tau \leq n- \tau$ does not restrict lists of exponential size to a certain code rate. However, the following remark shows anyway what happens if we assume $\tau > n-\tau$.

\begin{remark}[Case $\tau > n-\tau$]
Let $\dhalfnicefrac+1 \leq \tau < d\leq n\leq m$ and $\tau > n-\tau$. 
We can apply the same strategy as before: construct a constant-dimension code and show the existence of a constant-rank code of certain cardinality. 
For simplicity, consider only the case when $d$ is even, the case of odd $d$ follows immediately. 
Consider the lifting of a linear $\MRDlinear{n-\tau,n-\tau-d/2+1}$ code $\mycode{C}$ over $\mathbb{F}_{q^\tau}$ of minimum rank distance $d/2$. 
Now, let us lift $\mathcal{I}(\mycode{C})$, i.e., we consider $[\I_{\tau} \ \C_i]$ with $\C_i \in \Fq^{\tau \times (n-\tau)}$ for all $i=1,\dots,|\mycode{C}|$. 
In contrast to Lemma~\ref{lem:subspacecode_lifted_mrd}, we do not transpose the codewords of the MRD code here.
The subspaces defined by this lifting are a $\CDC{n,M_S,d_S=d,\tau}$ constant-dimension code of cardinality $M_S=q^{\tau(n-\tau-d/2+1)}$.

Then, with the same method as in Theorems~\ref{theo:exists_crc} and \ref{theo:lower_bound_list_size} and a $\CDC{m,|\mycode{M}|,d,\tau}$ code $\mycode{M}$ and a $\CDC{n,|\mycode{N}|,d,\tau}$ code $\mycode{N}$, there exists an $\code{n,M,d_R \geq d}$ code $\mycode{C}$ in rank metric and a word $\r \in \Fqm^n$
such that 
\begin{equation*}
\big|\mycode{C}\cap \Ball{\tau}{\r}\big| \geq q^{\tau(n-\tau-d/2+1)}.
\end{equation*}
However, the interpretation of this value is not so easy, since it depends on the concrete values of $\tau,d$ and $n$ if the exponent is positive and if this bound is exponential in $n$ or not.
Moreover, as mention before, we do not need this investigation for polynomial-time list decodability as Theorem~\ref{theo:lower_bound_list_size} shows that the list size is lower bounded by $q^{(n-\tau)}$ if we choose $\tau = \dhalfnicefrac+1$ for codes of \textbf{any} rate, where $\tau \leq n-\tau$ is fulfilled.
\end{remark}

The following lemma shows an improvement in the exponent of Theorem~\ref{theo:lower_bound_list_size} for the case $\tau = d/2$ or when $m$ is quite large compared to $n$.
\begin{lemma}[Bound of Theorem~\ref{theo:lower_bound_list_size} for $\tau = d/2$ or large $m$]\label{lem:refinement_lower_bound}
Let $\dhalfnicefrac < \tau < d < n$ and $\tau \leq n-\tau$. 
If either $\tau = d/2$ or $m \geq (n-\tau)(2\tau-d+1)+\tau+1$, then
there exists an $\code{n,M,d_R = d}$ code $\mycode{C}$ over $\Fqm$ of length $n\leq m$ and minimum rank distance $d$, and a word $\r \in \Fqm^n$
such that 
\begin{equation}\label{eq:lower_bound_rankmetric_refined}
\ell=\ell\big(m,n,d,\tau\big) \geq \big|\mycode{C}\cap \Ball{\tau}{\r}\big| \geq q^{(n-\tau)(2\tau-d+1)}.
\end{equation}
\end{lemma}
\begin{proof}
We use \cite[Theorem~2]{Gadouleau2010ConstantRank}, which shows that 
for $2r \leq n \leq m$ and $1 \leq \delta \leq r$
there exists a constant-rank code of cardinality
\begin{equation*}
\maxCardinalityRank{n,\delta+r,r} = \maxCardinalitySub{n,d_S=2\delta,r} 
\end{equation*}
if either $\delta = r$ or $m \geq (n-r)(r-d+1)+r+1$.

\begin{figure*}[ht!]
\normalsize
\begin{minipage}[b]{.47\textwidth}
\centering
\includegraphics{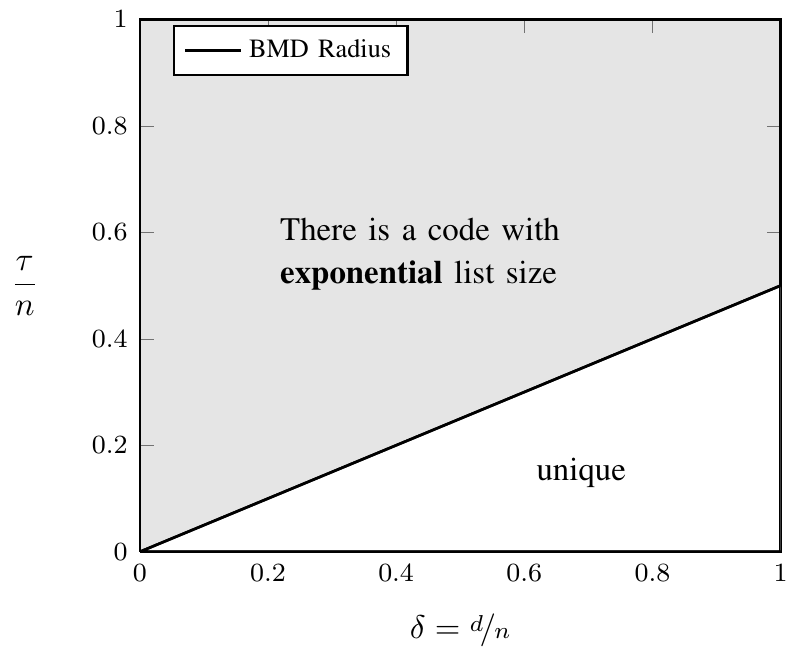}
\subcaption{General codes in rank metric}\label{fig:rank_decoding_region}
\end{minipage}
\begin{minipage}[b]{.47\textwidth}
\centering
\includegraphics{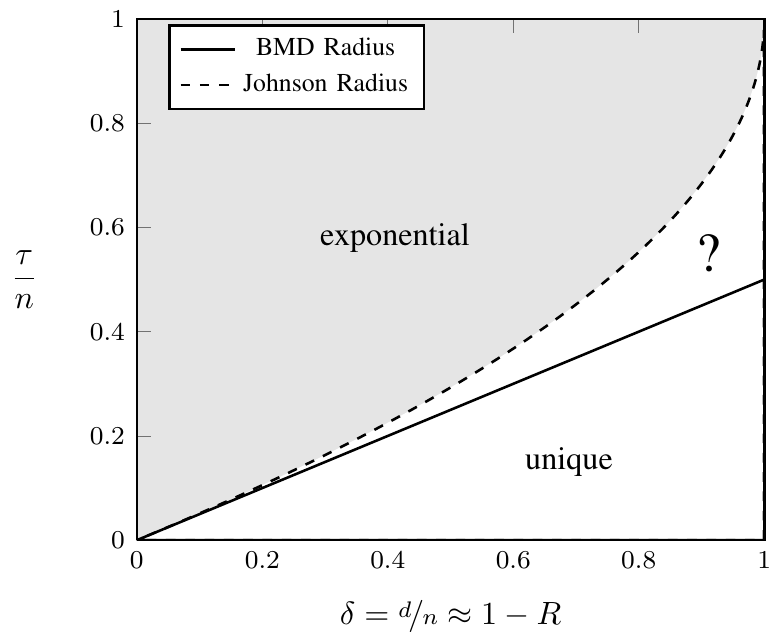}
\subcaption{Gabidulin codes}\label{fig:gabidulin_decoding_region}
\end{minipage}
\caption{List size of codes in rank metric, depending on normalized Bounded Minimum Distance (BMD) decoding radius $\tau_{BMD}/n = \dhalfnicefrac/n$ and normalized Johnson radius $\tau_J/n = \nicefrac{(\taujohnson)}{n}$
and on the normalized minimum distance $\delta = \nicefrac{d}{n}$. \label{fig:decoding_regions}}
\hrulefill
\end{figure*}

Thus, similar to the proof of Theorem~\ref{theo:upper_bound_listsize}, we choose
$r = \tau$ and $\delta = d-\tau$. Hence, there exists a $\CRC{n,M_R,d,\tau}$ constant-rank code of cardinality
\begin{align*}
M_R=\maxCardinalitySub{n,d_S=2(d-\tau),\tau} &\geq  q^{(n-\tau)(\tau-(d-\tau)+1)}\\
& = q^{(n-\tau)(2\tau-d+1)},
\end{align*}
where we used the cardinality of a constant-dimension code based on a lifted MRD code (see Lemma~\ref{lem:subspacecode_lifted_mrd}) as lower bound.
Analog to Theorem~\ref{theo:lower_bound_list_size}, we can use this constant-rank code to bound the list size.
\end{proof}
For the case $\tau = d/2$, this results in Corollary~\ref{cor:special_case_dhalf}.
Hence, for the cases of Lemma~\ref{lem:refinement_lower_bound}, the lower bound on the list size \eqref{eq:lower_bound_rankmetric_refined} and the upper bound \eqref{eq:upper_bound_rank} coincide up to a scalar factor and the upper bound is therefore asymptotically tight.

\section{Interpretation and Conclusion}\label{sec:interpretation}

This section interprets the results from the previous sections and compares them to known bounds on list decoding in Hamming metric (see e.g. \cite[Chapters 4 and 6]{Guruswami_ListDecodingofError-CorrectingCodes_1999}).

Theorem~\ref{theo:lower_bound_list_size} shows that there is a code over $\Fqm$ of length $n\leq m$ of rank distance at least $d$ 
such that there is a ball of any radius $\tau > \dhalfnicefrac$, which contains a number of codewords that grows exponentially in the length $n$. 
Hence, there exists a rank-metric code for which \emph{no} polynomial-time list decoding algorithm beyond half the minimum distance is possible. 
This bound is tight as a function of $d$ and $n$, since below we can clearly always decode uniquely. 
It does not mean that there exists no rank-metric code with a polynomial list size for a decoding radius greater than half the minimum distance, but 
in order to find a polynomial upper bound, it is necessary to use further properties of the code 
in the derivation of such bounds (linearity or the explicit code structure). 

In particular, for Gabidulin codes, there is still an unknown region between half the  minimum distance and the Johnson radius
since we could only prove that the list size can be exponential beyond the Johnson radius (see Theorem~\ref{theo:lower_bound_gabidulin}).
These decoding regions are shown in Fig.~\ref{fig:decoding_regions}, depending on the relative normalized minimum rank distance $\delta = d/n$.
%

Further, our lower bound from Theorem~\ref{theo:lower_bound_list_size} shows that there cannot exist a polynomial upper bound depending only on $n$ and $d$ similar to the Johnson bound for Hamming metric. Hence, our upper bound from Theorem~\ref{theo:upper_bound_listsize} differs only by a scalar factor of two in the exponent from the lower bound from Theorem~\ref{theo:lower_bound_list_size}. A shown in Lemma~\ref{lem:refinement_lower_bound}, the upper bound is even asymptotically tight in some cases.

These results show a surprising difference to codes in Hamming metric. \emph{Any} ball in Hamming metric of radius less than the Johnson radius $\tau_J = \taujohnson$ always contains a \emph{polynomial} number of codewords of \emph{any} code of length $n$ and minimum Hamming distance $d$. 
Moreover, it can be shown that there exist codes in Hamming metric with an exponential number of codewords if the radius is at least the Johnson radius \cite{Goldreich-Rubinfeld-Sudan:DM2000,Guruswami_ListDecodingofError-CorrectingCodes_1999}. However, it is not known whether this bound is also tight for special classes of codes, e.g. Reed--Solomon codes. 
This points out another difference between Gabidulin and Reed--Solomon codes, since for Reed--Solomon codes the minimum radius for which an exponential list size is proven is much higher \cite{Justesen2001Bounds,BenSasson2010Subspace} than for Gabidulin codes (see Theorem~\ref{theo:lower_bound_gabidulin}). 

Nevertheless, it is often believed that the Johnson bound is tight not only for codes in Hamming metric in general, but also for Reed--Solomon codes. Drawing a parallel conclusion for Gabidulin codes would mean that the maximum list size of Gabidulin codes could become exponential directly beyond half the minimum distance, but this requires additional research.

For future research, it is interesting to find a bound for the unknown region when list decoding Gabidulin codes. However, this seems to be quite difficult 
since the gap between the Johnson radius and the known lower exponential bounds for Reed--Solomon codes seems to translate into the
gap between half the minimum distance and the Johnson radius for Gabidulin codes
and despite numerous publications on this topic, nobody could close the gap for Reed--Solomon codes.
As a first step, it might be possible to prove something like Theorem~\ref{theo:lower_bound_list_size} for \emph{linear} codes in rank metric.

\section*{Acknowledgment}
The author thanks Pierre Loidreau for the valuable discussions and the careful reading of the manuscript, 
the reviewers for the helpful comments that helped to improve the presentation of the paper and Maximilien Gadouleau for pointing out the special case shown in Lemma~\ref{lem:refinement_lower_bound}.

%
%

\bibliographystyle{IEEEtran}
\bibliography{antoniawachter_bounds}
\end{document}